%% file: Class_HJ_sep_3.0_Main.tex
\documentclass[a4paper,11pt]{article} 
\usepackage{fancyhdr}
\usepackage{graphicx}
\usepackage{microtype} 

\usepackage[backend=bibtex,citestyle=alphabetic,bibstyle=alphabetic,backref=true,url=false]{biblatex}
\usepackage{amsfonts,amssymb,amstext,amsmath} 
\usepackage{cool} 
\usepackage[amsmath,thmmarks,hyperref]{ntheorem}
\usepackage{tensor}
\usepackage{tensind} 
\tensordelimiter{?}

\usepackage[bookmarksnumbered,hypertexnames=false]{hyperref} 
\hypersetup{
    plainpages=false,       
    unicode=false,          
    pdftoolbar=true,        
    pdfmenubar=true,        
    pdffitwindow=false,     
    pdfstartview={FitH},    
    pdftitle={Classification of Hamilton-Jacobi separation in orthogonal coordinates with diagonal curvature},    
    pdfauthor={Krishan Rajaratnam},     
    pdfnewwindow=true,      
    linktoc=page,           
    colorlinks=true,        
    linkcolor=blue,         
    citecolor=black,        
    filecolor=magenta,      
    urlcolor=cyan           
}

\numberwithin{equation}{section} 
\usepackage[capitalize]{cleveref} 
\input{commands2.0.tex}
\input{ntheorem.tex}
\renewcommand{\d}{{\textrm d}}

\newcommand{\ext}[1]{\overline{#1}}

\usepackage{autonum} 

\begin{document}
\pagenumbering{roman}
\title{Classification of Hamilton-Jacobi separation in orthogonal coordinates with diagonal curvature}
\author{Krishan Rajaratnam\footnote{e-mail: k2rajara@uwaterloo.ca}, Raymond G. McLenaghan\footnote{e-mail: rgmclenaghan@uwaterloo.ca}\\Department of Applied Mathematics, University of Waterloo, Canada}
\date{\today} 
\maketitle
\begin{abstract}\centering
	We find all orthogonal metrics where the geodesic Hamilton-Jacobi equation separates and the Riemann curvature tensor satisfies a certain equation (called the diagonal curvature condition). All orthogonal metrics of constant curvature satisfy the diagonal curvature condition. The metrics we find either correspond to a Benenti system or are warped product metrics where the induced metric on the base manifold corresponds to a Benenti system. Furthermore we show that most metrics we find are characterized by concircular tensors; these metrics, called Kalnins-Eisenhart-Miller (KEM) metrics, have an intrinsic characterization which can be used to obtain them. In conjunction with other results, we show that the metrics we found constitute all separable metrics for Riemannian spaces of constant curvature and de Sitter space.
\end{abstract}
\tableofcontents
\cleardoublepage
\phantomsection		


%

\pagenumbering{arabic}
\newpage
\include{Class_HJ_sep_3.0_Body}
\phantomsection
\addcontentsline{toc}{section}{References}
\printbibliography
\end{document}

%% file: commands2.0.tex
\usepackage{amsmath,amsfonts,amssymb}

\newcommand{\End}{\operatorname{End}}

\newcommand{\scalprod}[2]{\left\langle #1,#2 \right\rangle}

\def\sp(#1,#2){\left\langle #1,#2 \right\rangle}
\def\bp#1{\sp(#1)}
\newcommand{\norm}[1]{\left\| #1 \right\|} 

\DeclareFontFamily{U}{matha}{\hyphenchar\font45}
\DeclareFontShape{U}{matha}{m}{n}{
      <5> <6> <7> <8> <9> <10> gen * matha
      <10.95> matha10 <12> <14.4> <17.28> <20.74> <24.88> matha12
      }{}
\DeclareSymbolFont{matha}{U}{matha}{m}{n}
\DeclareFontFamily{U}{mathx}{\hyphenchar\font45}
\DeclareFontShape{U}{mathx}{m}{n}{
      <5> <6> <7> <8> <9> <10>
      <10.95> <12> <14.4> <17.28> <20.74> <24.88>
      mathx10
      }{}
\DeclareSymbolFont{mathx}{U}{mathx}{m}{n}

\DeclareMathSymbol{\obot}         {2}{matha}{"6B}
\DeclareMathSymbol{\bigobot}       {1}{mathx}{"CB}

\newcommand{\R}{\mathbb{R}}

\def\pderiv#1#2{\dfrac{\partial #1}{\partial #2}}

\def\spderiv#1#2#3{\dfrac{\partial^2 #1}{\partial #2\partial #3}}

\newcommand{\F}{{\mathcal{F}}}


%% file: ntheorem.tex
\usepackage{amsfonts,amssymb,amsmath}

\theoremnumbering{arabic}
\theoremstyle{break} 
\usepackage{latexsym}
\theoremsymbol{\ensuremath{_\Box}}
\theorembodyfont{\itshape}
\theoremheaderfont{\normalfont\bfseries}
\theoremseparator{}
\newtheorem{theorem}{Theorem}[section]
\newtheorem{proposition}[theorem]{Proposition}

\theoremsymbol{\ensuremath{\diamondsuit}} 

\theoremsymbol{\ensuremath{_\Box}}

\theorembodyfont{\upshape}

\newtheorem{remark}[theorem]{Remark}
\newtheorem{definition}[theorem]{Definition}

\newtheorem{claim}{Claim}[theorem]

\theoremstyle{nonumberplain}
\theoremheaderfont{\scshape}
\theorembodyfont{\normalfont}
\theoremsymbol{\ensuremath{_\blacksquare}}
\usepackage{amssymb}
\newtheorem{proof}{Proof}

\qedsymbol{\ensuremath{_\blacksquare}}

\numberwithin{equation}{section} 

\newcounter{partCounter}
\newenvironment{parts}{
	\begin{list}{\bfseries{}Case \arabic{partCounter}~}{\usecounter{partCounter}}
}{
	\end{list}
}

%% file: Class_HJ_sep_3.0_Body.tex
\section{Introduction}

In a previous article \cite{Rajaratnam2014a}, KEM coordinates (webs) were defined and shown to be (orthogonal) separable coordinates for the Hamilton-Jacobi equation. It was shown in that article that the converse is true for Riemannian spaces of constant curvature. Namely that all orthogonal separable coordinate systems in Riemannian spaces of constant curvature are KEM coordinates.

In \cite{Rajaratnam2014a}, the BEKM separation algorithm was introduced and shown to be a complete test of separability for natural Hamiltonians in KEM coordinates. This raises the problem of obtaining the KEM coordinates defined on a given pseudo-Riemannian manifold. Motivated by the case of spaces of constant curvature, in this article we obtain certain necessary but insufficient conditions on a coordinate system to be a KEM coordinate system. These conditions are precisely that the coordinate system be orthogonally separable with diagonal curvature (to be defined shortly). Our solution is sufficient to extend the Kalnins-Miller classification to spaces of constant curvature with arbitrary signature by giving an independent proof. In other words, our results will be sufficient to prove that every orthogonal separable coordinate system in a space of constant curvature is a KEM coordinate system.

In order to solve this problem, we note some hints from the literature. First we need a definition, a coordinate system is said to have \emph{diagonal curvature} if the Riemann curvature tensor satisfies $R_{ijik} = 0$ for $j \neq k$ in the coordinate induced basis. This definition is equivalent to requiring the curvature operator (which is a $\binom{2}{2}$-tensor associated with $R$ which induces a map in $\End(\wedge^{2}(M))$ \cite{Peter2006}) to be diagonal in the coordinate induced basis. Now, for the special case where the KEM coordinates are generated by a single concircular tensor with functionally independent eigenfunctions, Crampin has made some progress towards this problem in \cite{Crampin2003}. Indeed, proposition~6 in \cite{Crampin2003} states orthogonal separable coordinates with diagonal curvature satisfying an additional technical condition are generated by a concircular tensor.

Furthermore a careful study of the classification of Hamilton-Jacobi separation for Riemannian spaces of constant curvature due to Kalnins and Miller \cite{Kalnins1986} shows that a key to their solution was the diagonal curvature assumption. Their work is based on the seminal article by Eisenhart \cite{Eisenhart1934} where Eisenhart was able to make progress on the classification of orthogonally separable metrics assuming the coordinates had diagonal curvature (see \cite[Section 3]{Eisenhart1934}).

Motivated by these hints, we prove later that KEM coordinates have diagonal curvature in \cref{prop:KemDiagCurv}. Hence KEM coordinates are orthogonal separable coordinates with diagonal curvature. This gives a partial characterization of KEM coordinates. In this article we solve for all orthogonally separable metrics with diagonal curvature. This solution is sufficient to solve for all orthogonally separable metrics in spaces of constant curvature, thereby showing that all orthogonally separable metrics in these spaces are KEM metrics.

In the following subsections we will elaborate on such notions as KEM coordinates and present our main results in more detail.

\subsection{Concircular tensors and KEM coordinates}

We now define concircular tensors and KEM coordinates. First note that throughout this article, we assume $M$ is a $C^{\infty}$ manifold equipped with covariant pseudo-Riemannian metric $g$, (inverse) contravariant metric $G$ and the Levi-Civita connection induced by $g$ ia denoted by $\nabla$. A \emph{concircular tensor also called C-tensor}, $L$, is a symmetric contravariant tensor satisfying the following equation:

\begin{equation}
	\nabla_{k}L_{ij} = \alpha_{(i}g_{j)k}
\end{equation}

\noindent for some covector $\alpha$. See \cite{Rajaratnam2014a} and references therein for more on concircular tensors. By an \emph{orthogonal concircular tensor}, we mean a concircular tensor whose uniquely determined $\binom{1}{1}$-tensor is point-wise diagonalizable. It can be shown that for any OCT $L$ there exists a local product manifold $\prod_{i=0}^{k} M_{i}$ where $\dim M_{i} > 1$ for $i > 0$ such that the following hold \cite[theorem~6.1]{Rajaratnam2014a}:

\begin{itemize}
	\item If $F_{0}$ is the canonical foliation induced by $M_{0}$, then $L|_{F_{0}}$ has simple eigenfunctions. In particular, $L|_{M_{0}}$ induces separable coordinates $(x_{0})$ for $(M_{0}, g|_{M_{0}})$.
	\item If $F_{i}$ is the canonical foliation induced by $M_{i}$ for $i > 0$, then $L|_{F_{i}}$ is a constant multiple of $G|_{F_{i}}$.
\end{itemize}

An important property of this product manifold is the following: if $(x_{i})$ are separable coordinates on $(M_{i}, g|_{M_{i}})$ for $i > 0$ then the product coordinates $(x_{0},x_{1},\dotsc,x_{k})$ are separable coordinates for $(M,g)$ \cite{Rajaratnam2014a}. This observation motivates the following definition:

\begin{definition}[KEM coordinates \cite{Rajaratnam2014a}]
	Let $L$ be a non-trivial\footnote{By a non-trivial concircular tensor, we mean one which is not a multiple of the metric when $n > 1$.} OCT with associated product manifold $\prod_{i=0}^{k} M_{i}$ as above. For each $i= 1,...,k$, let $(x_{i})$ be KEM coordinates on $M_{i}$. Then the product coordinates $(x_{0},x_{1},\dotsc,x_{k})$ are called  \emph{Kalnins-Eisenhart-Miller (KEM) coordinates}.
\end{definition}
\begin{remark}
	When $L$ has simple eigenfunctions, the above definition is non-recursive and the KEM coordinates are separable coordinates associated with a Benenti tensor (by definition) \cite{Benenti2005a}. In particular when $n = 2$ every non-trivial orthogonal concircular tensor is a Benenti tensor. Hence $n = 2$ defines a base case for the above recursive definition.
\end{remark}

It was shown in \cite[proposition~6.7]{Rajaratnam2014a} that KEM coordinates are necessarily separable. We will show later on that KEM coordinates have diagonal curvature. The separability of these coordinates implies that the metric satisfies the Levi-Civita equations in these coordinates \cite{Levi-Civita1904}. Using these facts, we will solve for all orthogonally separable metrics with diagonal curvature.

We also note that the orthogonal separability of coordinates $(x^i)$ is characterized by the existence of a \emph{characteristic Killing tensor} diagonalized in these coordinates \cite{Benenti1997a}. To elaborate, first recall that a Killing tensor is a symmetric $2$-tensor $K$ satisfying the following equation

\begin{equation}
	\nabla_{(i}K_{jk)} = 0
\end{equation}

A Killing tensor $K$ is called characteristic if it has point-wise real simple eigenvalues and admits local coordinates in which it is diagonal. In conclusion, we note that the classification of orthogonally separable coordinates on a given pseudo-Riemannian manifold $M$ is equivalent to the classification of characteristic Killing tensors on $M$.


\subsection{Orthogonally separable metrics with diagonal curvature}

In this section we will present the results of this article in detail. But first we need a preliminary characterization of orthogonal concircular tensors. Suppose $(x_{i})$ are local coordinates and $L$ is a tensor defined as follows:

\begin{equation} \label{eq:CTeqn}
	L = \sum\limits_{a \in M} \sigma_{a} \partial_{a} \otimes \d x_{a} + \sum\limits_{I \in P} e_{I} \sum\limits_{i \in I}  \partial_{i} \otimes \d x_{i}
\end{equation}

\noindent where $\{1,\dotsc,n\} = M \cup (\cup_{I \in P} I)$ is a partition (here $P$ is an index set and each $I \in P$ is a subset of $\{1,\dotsc,n\}$), the $\sigma_{a}(x_{a})$ are non-constant and the $e_{I}$ are constants. It can be shown that if $L$ is a concircular tensor, then the metric has the following form \cite{Rajaratnam2014a} (cf. \cite[Section~18]{Benenti2005a}):

\begin{subequations} \label{eq:SCKTmetricStruct}
	\begin{equation} \label{eq:SCKTmetricStructA}
		g = \sum\limits_{a \in M} \Phi_{a} \prod\limits_{ \substack{  b \in M  \\ b \neq a}} (\sigma_{a} - \sigma_{b})\d x_{a}^{2} + \sum\limits_{I \in P} \left (\prod\limits_{a \in M} (e_{I} - \sigma_{a}) \right) g^{I}
	\end{equation}
	\begin{equation} \label{eq:SCKTmetricStructB}
		g^{I}_{ij} = \begin{cases}
			f^{I}_{ij}(x^{I}) & i,j \in I \\
			0 & i \notin I
		\end{cases}
	\end{equation}
\end{subequations}

\noindent where $\Phi_{a}$ is a function of $x_a$ only. Conversely, if the metric has the above form, then $L$ is a CT \cite{Rajaratnam2014a}. It will follow by the proof of our main result (see \cref{sec:classOsepD}), that given a metric with the above form, one can construct $L$ such that its eigenspaces are uniquely determined from the metric.

We will see that most orthogonally separable metrics with diagonal curvature have a form given by the above equation, i.e. they admit a concircular tensor diagonalized in the coordinates. We now list the general form of orthogonally separable metrics with diagonal curvature.

The ones having a form given by \cref{eq:SCKTmetricStruct} can be divided into the following three classes. The first class are the \emph{irreducible metrics}

\begin{equation} \label{eq:irredMetric}
	g = \sum\limits_{a = 1}^{n} \Phi_{a} \prod\limits_{ \substack{b \neq a}} (\sigma_{a} - \sigma_{b})\d x_{a}^{2}
\end{equation}

\noindent which occur when the eigenfunctions of any associated concircular tensor are functionally independent. These metrics were first found by Eisenhart in his article \cite{Eisenhart1934}.  The remaining two classes of metrics are referred to as reducible metrics. The following are \emph{product metrics}

\begin{equation} \label{eq:prodMetric}
	g = \sum\limits_{I = 1}^{p} g^{I}
\end{equation}

\noindent where each $g^{I}$ is given in \cref{eq:SCKTmetricStructB}. The final class are the \emph{warped product metrics}

\begin{equation} \label{eq:warpedProdMetric}
	g = \sum\limits_{a = 1}^{m} \Phi_{a} \prod\limits_{ \substack{  b \leq m  \\ b \neq a}} (\sigma_{a} - \sigma_{b})\d x_{a}^{2} + \sum\limits_{I = 1}^{p} \left (\prod\limits_{a \leq m} (e_{I} - \sigma_{a}) \right) g^{I}
\end{equation}

\noindent where each $g^{I}$ is given in \cref{eq:SCKTmetricStructB}.

There is one class of orthogonally separable metric with diagonal curvature which is not in general associated with a concircular tensor, it is given as follows:

 \begin{equation} \label{eq:iregMetric}
 	g = \Phi_{1} \d x_{1}^{2} + \sum\limits_{I = 1}^{p} \sigma_{1}^{I} g^{I}
 \end{equation}

\noindent where $\Phi_{1}, \sigma_{1}^{I}$ are functions of $x_{1}$ at most with each $\sigma_{1}^{I}$ non-constant. In conclusion, every orthogonally separable metric with diagonal curvature has a form given by \cref{eq:SCKTmetricStruct} or \cref{eq:iregMetric}. We will show later that if $g$ is an orthogonally separable metric with diagonal curvature, then each of the metrics $g^{I}$ must also be an orthogonally separable metric with diagonal curvature. This shows why the classification is recursive: if $|I| > 1$ then our classification will tell us that each $g^{I}$ must be of the form given by \cref{eq:SCKTmetricStruct} or \cref{eq:iregMetric}. Thus one must recursively apply this classification to obtain all orthogonally separable metrics with diagonal curvature for a given dimension.

Using the above classification, we will prove the following theorem concerning orthogonal separation in spaces of constant curvature:

\begin{theorem}[KEM Separation Theorem] \label{thm:KEMsep}
	Suppose $(M,g)$ is a space of constant curvature. In orthogonal separable coordinates, $g$ necessarily has the form given by \cref{eq:SCKTmetricStruct}.

	In terms of tensors, suppose $K$ is a characteristic Killing tensor defined on $M$. Then there is a non-trivial concircular tensor $L$ defined on $M$ such that each eigenspace of $K$ is L-invariant, i.e. $L$ is diagonalized in coordinates adapted to the eigenspaces of $K$. Furthermore, the eigenspaces of $L$ are uniquely determined by the separable web defined by $K$.
\end{theorem}

The above theorem is a generalization of the results due to Kalnins and Miller from \cite{Kalnins1986}; it holds in Lorentzian spaces as well. Together with theory presented in \cite[theorem~6.8]{Rajaratnam2014a}, the above theorem shows that in a space of constant curvature, every orthogonal separable coordinate system is a KEM coordinate system. Furthermore we note here that these conclusions together with the results presented in \cite{Benenti1992} (cf. \cite{Kalnins1986}) allow us to conclude the following:

\begin{theorem}[KEM Separation Theorem II]
	Suppose $(M,g)$ is a space of constant curvature with Euclidean signature or Lorentzian\footnote{We take Lorentzian signature to be $(-+\dots +)$} signature with positive curvature. Then every separable (not necessarily orthogonal) coordinate system has an orthogonal equivalent which is a KEM coordinate system.
\end{theorem}

We now provide an outline of the layout of this article. In \cref{sec:premRes} we first show that any KEM coordinates are necessarily orthogonal separable coordinates with diagonal curvature. Then we include the first steps of the derivation of all orthogonal separable coordinates with diagonal curvature; this part of the derivation can be found in the literature and so is included here for completeness. In \cref{sec:classOsepD} we finish off this derivation and consequently prove that the metrics listed above do in fact constitute all separable metrics with diagonal curvature. Finally in \cref{sec:oSepSCC} we will do additional calculations in order to prove the KEM separation theorem.

\section{Preliminary results} \label{sec:premRes}

In this section we introduce the problem we wish to solve precisely. Then we do some relevant calculations from the literature for completeness. The outline is as follows: First we show that any KEM coordinates are necessarily orthogonal separable coordinates with diagonal curvature. Then we solve for all metrics which satisfy these conditions. In this section we will partially complete this calculation using results from the literature and then finish it in the next section.

We first prove the following optional result which is included for completeness:

\begin{proposition} \label{prop:KemDiagCurv}
	KEM coordinates are orthogonal separable coordinates with diagonal curvature.
\end{proposition}
\begin{proof} 
	Assume that $(x^{i})$ are KEM coordinates. It follows from proposition~6.7 in \cite{Rajaratnam2014a} that these coordinates are separable. We now show that they necessarily have diagonal curvature. To do this, we use the eq.~(3) in \cite{Meumertzheim1999} which is a formula for the Riemann curvature tensor in a twisted product. Observe that the metric necessarily has the following form:
	
	\begin{equation}
		g = \sum\limits_{a \in M} e_{a}\rho_{a}^{2} \d x_{a}^{2} + \sum\limits_{I \in P} \rho_{I}^{2} g^{I}
	\end{equation}
	
	\noindent where $\{1,\dotsc,n\} = M \cup (\cup_{I \in P} I)$ is a partition (here $P$ is an index set and each $I \in P$ is a subset of $\{1,\dotsc,n\}$), each $e_{a} = \pm 1$ as the case may be and each $\rho_{i}(x_{1},\dotsc,x_{m})$ is a positive valued function and each $g^{I}$ is given by \cref{eq:SCKTmetricStructB}.
	
	Define $\tilde{g}$ as follows:
	
	\begin{equation} 
		\tilde{g} = \sum\limits_{a \in M} e_{a} \d x_{a}^{2} + \sum\limits_{I \in P} g^{I}
	\end{equation}
	
	Let $R$ (resp. $\tilde{R}$) denote the Riemann curvature tensor of $g$ (resp. $\tilde{g}$). Then for $i \in I, j \in J, k \in K$ with $i , j , k$ distinct, it follows from eq.~3 in \cite{Meumertzheim1999} that
	
	\begin{align}
		\scalprod{(R(\partial_{i}, \partial_{k}) - \tilde{R}(\partial_{i}, \partial_{k}))\partial_{j}}{\partial_{i}} & = g(\partial_{i},\partial_{i})(\bp{ \nabla_{\partial_{k}} U_I - \bp{\partial_{k}, U_I} U_I, \partial_{j} })
	\end{align}
	
	\noindent where $U_I = - \nabla \log \rho_I$ is the negative gradient of $\log \rho_I$. By using eq.~(2) in \cite{Meumertzheim1999}, we get the following:
	
	\begin{equation}
		\bp{ \nabla_{\partial_{k}} U_I - \bp{\partial_{k}, U_I} U_I, \partial_{j} } = (\partial_{k} \log \rho_J \partial_j + \partial_j \log \rho_K \partial_k) \log \rho_I - (\partial_k \log \rho_I) (\partial_j \log \rho_I)
	\end{equation}
	
	The above vanishes if either $j \notin M$ or $k \notin M$. So we can assume further that $I , J , K$ are distinct. Then from a direct calculation using the specific form of the twist functions (see \cite[theorem~6.1]{Rajaratnam2014a}), it follows that the above is identically zero in this case. Thus we have proven that if $i,j,k$ are distinct, then 
	\begin{equation}
		\scalprod{R(\partial_{i}, \partial_{k})\partial_{j}}{\partial_{i}} = \scalprod{\tilde{R}(\partial_{i}, \partial_{k})\partial_{j}}{\partial_{i}}
	\end{equation}

	First observe that $R_{ijik} = \scalprod{R(\partial_{i}, \partial_{k})\partial_{j}}{\partial_{i}}$ and we can assume $i,j,k$ are distinct to check the diagonal curvature condition. Also note that $\tilde{R}(\partial_{i}, \partial_{k})\partial_{j}$ is not necessarily zero only if $I = J = K$ or if $i,j,k \in M$ (see \cite[corollary~2]{Meumertzheim1999}). In the later case, clearly $R_{ijik} = 0$. In the former case the result follows by induction from the above equation.
\end{proof}

We now assume $(x^{i})$ are orthogonal separable coordinates with diagonal curvature. Assume the covariant metric $g = \operatorname{diag}(e_{1}H_{1}^{2},...,e_{n}H_{n}^{2})$ where each $e_{i} = \pm 1$ as the case may be. The separability of this metric implies it satisfies the Levi-Civita equations \cite{Levi-Civita1904,Benenti1997a}:


\begin{subequations} \label{IntConds}
	\begin{equation} \label{IntCondsA}
		\spderiv{\log H_{i}^{2}}{x_{i}}{x_{j}} + \pderiv{\log H_{i}^{2}}{x_{j}}\pderiv{\log H_{j}^{2}}{x_{i}} = 0
	\end{equation}
	\begin{equation} \label{IntCondsB}
		\spderiv{\log H_{i}^{2}}{x_{j}}{x_{k}} - \pderiv{\log H_{i}^{2}}{x_{j}}\pderiv{\log H_{i}^{2}}{x_{k}} + \pderiv{\log H_{i}^{2}}{x_{j}}\pderiv{\log H_{j}^{2}}{x_{k}} + \pderiv{\log H_{i}^{2}}{x_{k}}\pderiv{\log H_{k}^{2}}{x_{j}} = 0
	\end{equation}
\end{subequations}

\noindent where $i$,$j$, and $k$ are all distinct. We will now proceed to solve the above equations augmented with the diagonal curvature condition.

The following calculation is from \cite[proposition~6]{Crampin2003} which is adapted from Kalnins' book \cite{Kalnins1986} which is from \cite{Eisenhart1934}. First note that in orthogonal coordinates the Riemann curvature component $R_{jiik}$ for $i$,$j$,$k$ distinct has the following form \cite{Eisenhart1934}:

\begin{multline} \label{Rjiik}
	R_{jiik} = \frac{e_{i} H_{i}^{2}}{4}  \biggr [ 2 \spderiv{\log H_{i}^{2}}{x_{j}}{x_{k}} + \pderiv{\log H_{i}^{2}}{x_{j}}\pderiv{\log H_{i}^{2}}{x_{k}} \\
	- \pderiv{\log H_{i}^{2}}{x_{j}}\pderiv{\log H_{j}^{2}}{x_{k}} - \pderiv{\log H_{i}^{2}}{x_{k}}\pderiv{\log H_{k}^{2}}{x_{j}} \biggr ]
\end{multline}

In consequence of the second integrability condition, \cref{IntCondsB}, we find that:

\begin{equation} \label{RjiikSep}
	R_{jiik} = \frac{3 }{4} e_{i} H_{i}^{2}  \spderiv{\log H_{i}^{2}}{x_{j}}{x_{k}}
\end{equation}

Thus the diagonal curvature assumption implies that for $i$,$j$,$k$ distinct:

\begin{equation} \label{eq:KemCond}
	\spderiv{\log H_{i}^{2}}{x_{j}}{x_{k}} = 0
\end{equation}

Solving the above equation we find that:

\begin{equation} \label{genHForm1}
	H_{i}^{2} = \prod \limits_{j \neq i} \Psi_{i j}(x_{i}, x_{j})
\end{equation}

Now the first integrability condition, \cref{IntCondsA}, applied twice to $i \neq j$ implies that:

\begin{subequations}
	\begin{equation} \label{IntCondsA1}
		\spderiv{\log H_{i}^{2}}{x_{i}}{x_{j}} = -\pderiv{\log H_{i}^{2}}{x_{j}}\pderiv{\log H_{j}^{2}}{x_{i}}
	\end{equation}
	\begin{equation} \label{IntCondsA2}
		\spderiv{\log H_{i}^{2}}{x_{i}}{x_{j}}  = \spderiv{\log H_{j}^{2}}{x_{i}}{x_{j}}
	\end{equation}
\end{subequations}

If we substitute the form of $H$ from \cref{genHForm1} into  \cref{IntCondsA2} we have that:

\begin{equation}
	\spderiv{\log \dfrac{\Psi_{ij}}{\Psi_{ji}} }{x_{i}}{x_{j}}  = 0
\end{equation}

Thus

\begin{equation}
\dfrac{\Psi_{ij}}{\Psi_{ji}} = \dfrac{\chi_{ij}(x_{i})}{\chi_{ji}(x_{j})}
\end{equation}

If we let $\Phi_{ij} = \Phi_{ji} = \dfrac{\Psi_{ij}}{\chi_{ij}} $ and $\Phi_{i} = \prod\limits_{j \neq i} \chi_{ij} $, \cref{genHForm1} becomes:

\begin{equation} \label{genHForm2}
	H_{i}^{2} = \Phi_{i}(x_{i}) \prod \limits_{j \neq i} \Phi_{i j}(x_{i}, x_{j})
\end{equation}

Now if we substitute the form of $H$ above into \cref{IntCondsA1} we have that:

\begin{equation}
	\spderiv{\Phi_{ij} }{x_{i}}{x_{j}}  = 0
\end{equation}

Thus 

\begin{equation}
	\Phi_{ij}(x_{i},x_{j}) = \sigma_{ij}(x_{i}) + \sigma_{ij}(x_{j})
\end{equation}

This gives us the following general form of $H$ satisfying \cref{IntCondsA} and \cref{eq:KemCond}:

\begin{equation} \label{HForm}
	H_{i}^{2} = \Phi_{i}(x_{i})\prod_{j \neq i} (\sigma_{ij}(x_{i}) + \sigma_{ji}(x_{j})) \quad (i=1,..,n) 
\end{equation}

The above equation was first derived by Eisenhart in his seminal paper \cite{Eisenhart1934} and it was used resourcefully by Kalnins and Miller in their classification of separable coordinates systems in $S^{n}, E^{n}$ and $H^{n}$ \cite{Kalnins1986}. When $n=2$, the above equation gives the general solution and it follows that the metric has the form given by \cref{eq:SCKTmetricStruct}. Thus for the remainder of the article we assume $n > 2$. Now for $i$,$j$,$k$ distinct we evaluate \cref{IntCondsB} with all cyclic permutations of $i$,$j$,$k$ using the form of $H$ given above to get the following system of equations:

\begin{subequations} \label{eqn:remain}
	\begin{equation} \label{eqn:remainA}
		\sigma_{ji}'\sigma_{ki}'(\sigma_{jk} + \sigma_{kj}) - \sigma_{ji}'\sigma_{kj}'(\sigma_{ki} + \sigma_{ik}) - \sigma_{ki}'\sigma_{jk}'(\sigma_{ij} + \sigma_{ji}) = 0
	\end{equation}
	\begin{equation} \label{eqn:remainB}
		\sigma_{kj}'\sigma_{ij}'(\sigma_{ki} + \sigma_{ik}) - \sigma_{kj}'\sigma_{ik}'(\sigma_{ij} + \sigma_{ji}) - \sigma_{ij}'\sigma_{ki}'(\sigma_{jk} + \sigma_{kj}) = 0
	\end{equation}
	\begin{equation} \label{eqn:remainC}
		\sigma_{ik}'\sigma_{jk}'(\sigma_{ij} + \sigma_{ji}) - \sigma_{ik}'\sigma_{ji}'(\sigma_{jk} + \sigma_{kj}) - \sigma_{jk}'\sigma_{ij}'(\sigma_{ki} + \sigma_{ik}) = 0
	\end{equation}
\end{subequations}

\noindent where the primes indicates differentiation. Now since each $\Phi_{ij} = \sigma_{ij} + \sigma_{ji} $ is non-zero, the determinant of the above equations must vanish, this gives us the following equation:

\begin{equation} \label{eq:remainDet}
	\sigma_{ij}'\sigma_{jk}'\sigma_{ki}' + \sigma_{ji}'\sigma_{kj}'\sigma_{ik}' = 0
\end{equation}

We will solve the remaining equations in the next section.

\section{Classification of orthogonal separable coordinates with diagonal curvature} \label{sec:classOsepD}

We continue the derivation started in the previous section. An important subset of coordinates are the coordinates $i$ which satisfy $\sigma_{ij}' \neq 0 \quad \forall j \neq i$. These coordinates will be called \emph{connecting coordinates} for reasons that will become apparent later on. The set of all connecting coordinates for a given separable metric will be denoted by $M$ and we will assume the coordinates are chosen such that $M = \{1,...,m\}$.

First we give a rough idea of how we will do this classification. When there are no connecting coordinates, we show that metric is necessarily a product metric. When there is at least one connecting coordinate we show that the metric is any one of the other metrics listed in the introduction. In order to prove that the metric is a product metric when it has no connecting coordinates we define a relation among the coordinates. We then prove that this relation is an equivalence relation. Then we use this equivalence relation to prove that the metric has at least one connecting coordinate or is a product metric.

We now define a relation among the coordinates to distinguish between the different possible metrics that can occur. The relation is designed so that if it gives multiple partitions then these partitions are associated with a product metric. Furthermore, we should be able to conclude that the metric is connected if there is only one partition. It's easiest to first define when two coordinates $i$ and $j$ are inequivalent. If $I$ and $J$ are distinct partitions from the product metric in \cref{eq:prodMetric} and $i \in I$ and $j \in J$, then the first thing to notice is that $\sigma_{ij}' = \sigma_{ji}' = 0$. But with this definition of in-equivalence, if there are multiple partitions, it's still possible that we're dealing with a warped product metric given by \cref{eq:warpedProdMetric}; we need to make sure that there is no third coordinate $k$ such that $\sigma_{ki}',  \sigma_{kj}' \neq 0$. This gives us a definition of equivalence:

\begin{definition} \label{defn:connected}
	Two distinct variables $i$ and $j$ are said to be connected and denoted $i \sim j$ if one of the following conditions hold:
	\begin{align}
		\sigma_{ij}' \neq 0 \\
		\sigma_{ji}' \neq 0 \\
		\exists \; k \neq i,j : \sigma_{ki}', \sigma_{kj}' \neq 0
	\end{align}
	
	Also we define $\sim$ such that $i \sim i$.
\end{definition}

There are two special types of connectedness that arise, the first is when $i$,$j$ satisfy $\sigma_{ij}' \neq 0$ or $\sigma_{ji}' \neq 0$, in this case we say that $i$ and $j$ are \emph{strongly connected}. If $i$,$j$ are connected but not strongly connected, we say that $i$ and $j$ are \emph{weakly connected} by $k$ or just that $i$ and $j$ are \emph{weakly connected}.

\begin{proposition}
	The relation $\sim$ defined in Definition~\ref{defn:connected} is an equivalence relation.
\end{proposition}

\begin{proof}
	We check that this relation is transitive, as reflexivity and symmetry are immediately verified. So suppose that $i \sim j$ and $j \sim k$ where $i$,$j$,$k$ are mutually distinct.
	\begin{parts}
		\item ($\sigma_{ji}' \neq 0 \text{ and }  \sigma_{jk}' \neq 0$) \\
				In this case $i$ and $k$ are weakly connected by $j$.
	
		\item ($\sigma_{ij}' \neq 0 \text{ and }  \sigma_{kj}' \neq 0$) \\
			Assume to the contrary that $\sigma_{ik}'=\sigma_{ki}'=0$, then \cref{eqn:remainB} can't be satisfied. Thus $i$ must be strongly connected to $k$.
	
		\item ($\sigma_{ij}' \neq 0 \text{ and }  \sigma_{jk}' \neq 0$ or $\sigma_{ji}' \neq 0 \text{ and }  \sigma_{kj}' \neq 0$) \\
		Assume first that $\sigma_{ij}' \neq 0 \text{ and }  \sigma_{jk}' \neq 0$ and to the contrary that $\sigma_{ik}'=0$, then \cref{eqn:remainC} can't be satisfied. Also the case where $\sigma_{ji}' \neq 0 \text{ and }  \sigma_{kj}' \neq 0$ is just a permutation of the first, so the same argument applies. Thus in either case $i$ must be strongly connected to $k$.
		
		\item ($\sigma_{ij}' \neq 0$ and $j$ and $k$ are weakly connected or $\sigma_{kj}' \neq 0$ and $i$ and $j$ are weakly connected) \\
		Suppose first that $\sigma_{ij}' \neq 0$ and $j$ and $k$ are weakly connected by $h$. So we have that $\sigma_{hj}',\sigma_{hk}' \neq 0$.
		
		If $h = i$ then $\sigma_{ik}' \neq 0$, so assume that $h \neq i$. If $\sigma_{hi}' \neq 0$ then $i$ and $k$ are weakly connected by $h$, so assume that $\sigma_{hi}' = 0$. If $\sigma_{ih}' \neq 0$ then by Case 3 we get that $\sigma_{ik}' \neq 0$, so assume further that $\sigma_{ih}' = 0$. Then after checking \cref{eqn:remainC} with the following coordinates, we get a contradiction. 
		\begin{align}
		h \rightarrow i \\
		i \rightarrow j \\
		j \rightarrow k
		\end{align}
		The case where $\sigma_{kj}' \neq 0$ and $i$ and $j$ are weakly connected is just a permutation of the first case. Thus we conclude that $i$ is connected to $k$.
		
		\item ($\sigma_{ji}' \neq 0$ and $j$ and $k$ are weakly connected or $\sigma_{jk}' \neq 0$ and $i$ and $j$ are weakly connected) \\
		Suppose first that $\sigma_{ji}' \neq 0$ and $j$ and $k$ are weakly connected by $h$. So we have that $\sigma_{hj}',\sigma_{hk}' \neq 0$.
				
				If $h = i$ then $\sigma_{ik}' \neq 0$, so assume that $h \neq i$. Since $\sigma_{hj}' \neq 0$ and $\sigma_{ji}' \neq 0$, by Case 3 we get that $\sigma_{hi}' \neq 0$. Thus $i$ and $k$ are weakly connected by $h$.
				
				The case where $\sigma_{jk}' \neq 0$ and $i$ and $j$ are weakly connected is just a permutation of the first case. Thus we conclude that $i$ is connected to $k$.
		
		\item ($i$ and $j$ are weakly connected and $j$ and $k$ are weakly connected) \\
		Suppose l satisfies $\sigma_{li}',\sigma_{lj}' \neq 0$ and $h$ satisfies $\sigma_{hj}',\sigma_{hk}' \neq 0$.
		
		If $h = l$ then $i$ and $k$ are clearly weakly connected, so assume that $h \neq l$.
		
		Note that $l$ is strongly connected to $j$ and $j \sim k$ then we can use one of the previous cases considered to find that $l \sim k$. Similarly because $i$ is strongly connected to $l$ and $l \sim k$ we find that $i \sim k$.
	\end{parts}
	
	Thus we conclude that $\sim$ is transitive and thus defines an equivalence relation.
\end{proof}

Now suppose that $\sim$ gives a single partition of the coordinates, i.e. the coordinates are connected. Our goal is to show that there must be at least one connecting coordinate. First we need a definition. We define $S$, called the set of strongly connected coordinates as follows:

\begin{equation}
	S \equiv \{\; \text{ $i$ : $i$ is strongly connected to every $j$}  \}
\end{equation}

The reason to make this definition is because $M \subseteq S$ (this inclusion might be proper in some cases which can be observed by inspecting KEM metrics derived by Kalnins-Miller \cite{Kalnins1986}). So the idea is to first show that $S \neq \emptyset$ since this is easier to do using the hypothesis of connectedness. It turns out that this is possible.

\begin{proposition}
	When the coordinates are connected, $S$ has at least one coordinate.
\end{proposition}

\begin{proof}
	Suppose to the contrary that $S = \emptyset$. Then $\forall  i$, there exists $j,k$ with $ i,j,k$ distinct such that
	
	\begin{equation} \label{eq:noStrongCon}
		\sigma_{ij}'=\sigma_{ji}'=0 \text{ and } \sigma_{ki}',\sigma_{kj}' \neq 0
	\end{equation}
	
	So fix some $i$, and choose $j,k$ satisfying \cref{eq:noStrongCon}. Then by \cref{eqn:remainA} we must have that $\sigma_{jk}'=0$, similarly by \cref{eqn:remainB} we must have that $\sigma_{ik}'=0$. Now let $A = \{i,j\}$ then note that $k \notin A$ and $\forall \: l \in A, \:\: \sigma_{lk}' = 0$.
	
	\begin{claim}
		Suppose we have a coordinate $f$ and a set of coordinates $A \neq \emptyset$ such that $f \notin A$ and $\forall i \in A, \:\: \sigma_{if}' = 0$. Furthermore assume that $\{f\} \cup A \neq \{1,...,n\}$. Then we can obtain a new set $A'$ such that $A \cup \{f\} \subseteq A'$ and an $h \notin A'$ such that $\forall i \in A', \:\: \sigma_{ih}' = 0$.
	\end{claim}
	\begin{proof}[Proof of claim]
		By assumption there exists $g,h$ satisfying \cref{eq:noStrongCon} with $f$ in place of $i$ and $g$ in place of $j$. Since $\sigma_{hf}' \neq 0$, $h \notin A$. If $\{f\} \cup A = \{1,...,n\}$, then we have reached a contradiction, so assume otherwise. As we observed earlier for a similar case, we must have $\sigma_{gh}'=\sigma_{fh}'=0$. Also $\forall i \in A$ since $\sigma_{if}'=0$ and $\sigma_{hf}' \neq 0$ by evaluating \cref{eqn:remainB} with $i \rightarrow i, f \rightarrow j, h \rightarrow k$ we find that $\sigma_{ih}'=0$.
			
		Thus if we let $A' = A \cup \{g,f\}$ then $\forall  i \in A', \sigma_{ih}'=0$. Also note that $|A'| > |A|$ and $h \notin A'$.
	\end{proof}
	
	Now we can inductively apply Claim 1 to get a set of coordinates $A \neq \emptyset$, an $f \notin A$ such that $\forall i \in A, \:\: \sigma_{if}' = 0$ and $\{f\} \cup A = \{1,...,n\}$. Then by assumption there must exist a coordinate $g$ such that $f$ is weakly connected to $g$. So there is a coordinate $h$, with $h \neq f$, such that $\sigma_{hf}' \neq 0$. Since $\{f\} \cup A = \{1,...,n\}$, $h \in A$, thus $\sigma_{hf}' = 0$, a contradiction.

	Thus $S \neq \emptyset$.
\end{proof}

Then assuming $S \neq \emptyset$ we try to prove that $M \neq \emptyset$. This is also possible.

\begin{proposition} \label{prop:ConCoords}
	When the coordinates are connected and $S$ has at least one coordinate then there must be at least one connecting coordinate. Thus due to the previous proposition we find that when the coordinates are connected there must be at least one connecting coordinate.
\end{proposition}

\begin{proof}
	Assume to the contrary that $M = \emptyset$. Then $\forall i \in S$ there exists $j$ such that
	
	\begin{equation} \label{eq:noSimpleCoords}
	\sigma_{ij}' = 0 \text{ and } \sigma_{ji}' \neq 0
	\end{equation}
	
	Since $S \neq \emptyset$ by hypothesis, we can choose some $i \in S$ and some $j \neq i$ such that \cref{eq:noSimpleCoords} is satisfied. Let $B = \{i\}$.
	
	\begin{claim}
		Suppose $\emptyset \neq B \subseteq S$ and there is a $j \notin B$ such that $\forall i \in B$, $\sigma_{ij}'=0 $. Furthermore assume that $\{j\} \cup B \neq \{1,...,n\}$. Then we can obtain a new set $B' = \{j\} \cup B \subseteq S$ and a $k \notin B'$ such that $\forall i \in B', \:\: \sigma_{ik}'=0$.
	\end{claim}
	
	\begin{proof}[Proof of claim]
		Fix an $i \in B$, then $\sigma_{ij}'=0$ and $\sigma_{ji}' \neq 0$. Now pick $k \neq i,j$, then $\sigma_{ik}' \neq 0$ or $\sigma_{ki}' \neq 0$.
		
		If $ \sigma_{ik}' \neq 0 $ then by \cref{eqn:remainC} we must have that $ \sigma_{jk}' \neq 0 $. If $ \sigma_{ki}' \neq 0 $ then by \cref{eqn:remainA} either $ \sigma_{jk}' \neq 0 $ or $ \sigma_{kj}' \neq 0 $. In either case we find that $j$ is strongly connected to $k$. Since $k$ was arbitrary and because $j$ is also strongly connected to $i$ we find that $j \in S$. Then by \cref{eq:noSimpleCoords} there exists $\exists \: k \neq i,j$ such that $\sigma_{jk}'=0$ and $\sigma_{kj}' \neq 0$, note that $k \notin B$.
		
		Assume $i \in B$ is arbitrary, then $\sigma_{ij}'=0$ and $\sigma_{kj}' \neq 0$, thus by \cref{eqn:remainB} we must have that $\sigma_{ik}'=0$. 
		
		Let $B' = B \cup \{j\} \subseteq S$ then $\forall i \in B'$ we have $\sigma_{ik}'=0$, also note that $k \notin B'$.
	\end{proof}
	
	Now we can inductively apply Claim 1 to a get set $B$ satisfying $\emptyset \neq B \subseteq S$ and a $j \notin B$ such that $\forall i \in B$, $\sigma_{ij}'=0 $ and $\{j\} \cup B = \{1,...,n\}$. As in the proof of the Claim 1, we find that $j \in S$. Then by \cref{eq:noSimpleCoords} $\exists \: k \neq j$ such that $\sigma_{kj}' \neq 0$, but since $\{j\} \cup B = \{1,...,n\}$ we must have that $k \in B$, then $\sigma_{kj}' = 0$, a contradiction.
	
	Thus $M \neq \emptyset$.
\end{proof}

The following proposition classifies all metrics with at least one connecting coordinate.

\begin{proposition} \label{prop:OneConCoord}
	If the metric has at least one connecting coordinate then the following statements are true. For $a \in M$:
	
	\begin{equation}
		H_{a}^{2} = \Phi_{a} \prod\limits_{ \substack{  b \in M  \\ b \neq a}} (\sigma_{a} -\sigma_{b})
	\end{equation}
	
	If $m > 1$ then one can partition the coordinates in\footnote{If $Y \subseteq X$, then we denote the complement of $Y$ in $X$ (elements of $X$ not in $Y$) as $Y^{c}$.} $M^{c}$ such that if $I$ is an equivalence class of this partition and $\alpha \in I$, then

	\begin{align}
		H_{\alpha}^{2} = \Phi_{\alpha} \prod\limits_{\substack{\beta \in I \\ \beta \neq \alpha}} (\sigma_{\alpha \beta} + \sigma_{\beta \alpha}) \prod\limits_{a =1}^{m} (e_{I} -\sigma_{a}) \quad (m \geq 2)
	\end{align}
	
	If $m = 1$ then one can partition the coordinates in $M^{c}$ such that if $I$ is an equivalence class of this partition and $\alpha \in I$, then

	\begin{align}
		H_{\alpha}^{2} = \Phi_{\alpha} \sigma_{a}^{I} \prod\limits_{\substack{\beta \in I \\ \beta \neq \alpha}} (\sigma_{\alpha \beta} + \sigma_{\beta \alpha})\quad (m = 1)
	\end{align}
	
	Equations~\eqref{eqn:remain} are satisfied whenever at least one of $i$, $j$ or $k$ is in $M$ if and only if the functions $H_{i}^{2}$ are of the form just described. Furthermore Equations~\eqref{eqn:remain} are satisfied whenever $i,j,k$ are not all in the same partition.
\end{proposition}

\begin{proof}
	By hypothesis we can assume $m \geq 1$. We use Latin letters such as $a$ to denote the connecting coordinates and the remaining coordinates are denoted with Greek letters such as $\alpha$. Although $i$ and $j$ are reserved for arbitrary coordinates. Furthermore we denote $N = \{1,...,n\}$. Then by definition $\forall \: a \in M , i \in N$ we have that $\sigma_{a i}' \neq 0$.
	
	\begin{claim}
		For $a \in M$ and $\alpha \in M^{c}$, $\sigma_{\alpha a}' = 0$.
	\end{claim}
	\begin{proof}
		For any $\alpha \in M^{c}$, there exists $i \in N$ such that $\sigma_{\alpha i}' = 0$.
		
		Suppose first that $i \in M$ and let $a = i$. Suppose to the contrary that there exists $b \in M \setminus \{a\}$ such that $\sigma_{\alpha b}' \neq 0$. Then \cref{eqn:remainB} can't hold with $\alpha \rightarrow i, a \rightarrow j, b \rightarrow k$. Thus the claim holds in this case.
		
		If $i \in M^{c}$, let $\beta = i$. If the first case doesn't hold for $\alpha$, then $\sigma_{\alpha a}' \neq 0$ $\forall \: a \in M$. Fix $a \in M$, then \cref{eqn:remainB} can't hold with $\alpha \rightarrow i, \beta \rightarrow j, a \rightarrow k$. Thus the first case must hold for some $a \in M$, thus the claim must hold.
	\end{proof}
	
	The proof for the following claim is mainly from \cite[P.~292]{Eisenhart1934}.
	\begin{claim}
		For $a \in M$, the following holds
		\begin{equation}
			H_{a}^{2} = \Phi_{a} \prod\limits_{ \substack{  b \in M  \\ b \neq a}} (\sigma_{a} -\sigma_{b})
		\end{equation}
		where each $\sigma_{a}(x_{a})$.
	\end{claim}
	\begin{proof}
		Suppose first that $m = 1$ and let $a \in M$. Then for $\alpha \in M^{c}$ by the above claim we know that $\sigma_{a \alpha} + \sigma_{\alpha a}$ only depends on the $a$ coordinate and so these factors can be absorbed into $\Phi_{a}$ and $\Phi_{\alpha}$. Thus the claim holds in this case.
		
		So for the remainder of the proof of this claim assume that $m > 1$. To prove this statement, for $a, b \in M$ our goal is to remove the $b$ dependence from $\sigma_{ab}$. First assume $m > 2$ and let $a,b,c \in M$. From \cref{eq:remainDet} evaluated with $a \rightarrow i, b \rightarrow j, c \rightarrow k$ we get:
		
		\begin{equation}
			\sigma_{ab}'\sigma_{bc}'\sigma_{ca}' + \sigma_{ba}'\sigma_{cb}'\sigma_{ac}' = 0
		\end{equation}
		
		Since each term is non-zero, by separating variables it follows that $\dfrac{\sigma_{ab}'}{\sigma_{ac}'}$ is a constant. Thus we can set $\sigma_{ab} = a_{ab} \sigma_{a}$ where $a_{ab}$ is a constant and $\sigma_{a}$ involves $x_{a}$ at most. The above equation implies that the constants must satisfy the following:
		
		\begin{equation} \label{eq:constsCond}
			a_{ab}a_{bc}a_{ca} + a_{ba}a_{cb}a_{ac} = 0
		\end{equation}
		
		Assuming the above equation holds, it follows that all three Equations~\eqref{eqn:remain} are satisfied for $a \rightarrow i, b \rightarrow j, c \rightarrow k$. Now set
		
		\begin{align}
			\sigma_{a} &= a_{bc}a_{ca} \ext{\sigma}_{a} & \sigma_{b} &= a_{cb}a_{ac} \ext{\sigma}_{b}
		\end{align}
		
		Then \cref{eq:constsCond} implies that $a_{ab} \sigma_{a} + a_{ba}\sigma_{b} = a_{ab}a_{bc}a_{ca}(\ext{\sigma}_{a} - \ext{\sigma}_{b})$; in which case the constant factor may be absorbed into $\Phi_{a}$ and $\Phi_{b}$. Thus we can assume $a_{ab} = - a_{ba} = 1$, then \cref{eq:constsCond} becomes:
		
		\begin{equation}
			a_{bc}a_{ca} - a_{cb}a_{ac} = 0
		\end{equation}
		
		Now set $a_{ca} \sigma_{c} = - a_{ac} \ext{\sigma}_{c}$, then we have:
		
		\begin{equation}
			a_{ca} \sigma_{c} + a_{ac} \ext{\sigma}_{a} = a_{ac}(\ext{\sigma}_{a} - \ext{\sigma}_{c})
		\end{equation}
		
		Thus we can assume $a_{ac} = - a_{ca} = 1$. Then $a_{bc} \ext{\sigma}_{b} + a_{cb} \ext{\sigma}_{c} = a_{bc} (\ext{\sigma}_{b} - \ext{\sigma}_{c})$ and so we can assume $a_{bc} = - a_{cb} = 1$. Inductively, this process can be continued so that each $\sigma_{ab} = \pm \sigma_{a}$, where the sign is positive if $\sigma_{ab}$ appears in $H_{a}^{2}$ and is negative if $\sigma_{ab}$ appears in $H_{b}^{2}$.
		
		If $m = 2$ then we can define $\sigma_{a} = \sigma_{ab}$ and $\sigma_{b} = - \sigma_{ba}$ without loss of consistency. For $\alpha \in M^{c}$, $\sigma_{\alpha a}' = 0$, so we can absorb terms of the form $\sigma_{\alpha a} + \sigma_{a \alpha}$ into $\Phi_{a}$ and $\Phi_{\alpha}$. Thus we have proven the following:
		
		\begin{equation}
			H_{a}^{2} = \Phi_{a} \prod\limits_{ \substack{  b \leq m  \\ b \neq a}} (\sigma_{a} -\sigma_{b})
		\end{equation}
	\end{proof}
	
	We can now assume that Equations~\eqref{eqn:remain} have been solved whenever $i,j,k \in M$. Thus if $m = n$ the above claim proves that the metric has the form given by \cref{eq:irredMetric} and so we are finished. So assume for the remainder of the proof that $m < n$.
	
	Now fix $a,b \in M$ and $\alpha \in M^{c}$. Let $a_{\alpha a} = \sigma_{\alpha a} \in \R$ and $a_{\alpha b} = \sigma_{\alpha b} \in \R$. Then \cref{eqn:remainC} evaluated with $a \rightarrow i, b \rightarrow j, \alpha \rightarrow k$ gives:
	
	\begin{equation}
		\sigma_{a\alpha}'\sigma_{b\alpha}'(\sigma_{a} - \sigma_{b}) + \sigma_{a\alpha}'\sigma_{b}'(\sigma_{b\alpha} + a_{\alpha b}) - \sigma_{b\alpha}'\sigma_{a}'(a_{\alpha a} + \sigma_{a\alpha}) = 0
	\end{equation}
	
	We now proceed to solve the above equations. First we rearrange terms to separate the variables:
	
	\begin{align}
		(\sigma_{a} - \sigma_{b}) + \frac{\sigma_{b}'}{\sigma_{b\alpha}'}(\sigma_{b\alpha} + a_{\alpha b}) - \frac{\sigma_{a}'}{\sigma_{a\alpha}'}(a_{\alpha a} + \sigma_{a\alpha}) = 0 \label{eq:constMultpTemp} \\
		\Rightarrow \sigma_{a} - \frac{\sigma_{a}'}{\sigma_{a\alpha}'}(a_{\alpha a} + \sigma_{a\alpha}) = \sigma_{b} - \frac{\sigma_{b}'}{\sigma_{b\alpha}'}(\sigma_{b\alpha} + a_{\alpha b}) = c \in \R
	\end{align}
	
	Then one can show that $\dfrac{\sigma_{a}'}{\sigma_{a\alpha}'} = - d \in \R \setminus \{0\}$ and similarly $\dfrac{\sigma_{b}'}{\sigma_{b\alpha}'} = -f \in \R \setminus \{0\}$. Thus the above equation implies:
	
	\begin{align}
		\sigma_{a} & = - d (a_{\alpha a} + \sigma_{a \alpha} - \frac{c}{d}) \\
		\sigma_{b} & = - f (a_{\alpha b} + \sigma_{b \alpha} - \frac{c}{f})
	\end{align}
	
	Now let $e_{\alpha} = c$ then the two equations above implies the following:
	
	\begin{align}
		a_{\alpha a} + \sigma_{a \alpha} & = \frac{e_{\alpha} - \sigma_{a}}{d} \\
		a_{\alpha b} + \sigma_{b \alpha}  & = \frac{e_{\alpha} - \sigma_{b}}{f}
	\end{align}
	
	Thus by absorbing the constants $d,f$ into the $\Phi$ functions we can assume $a_{\alpha a} = a_{\alpha b} = e_{\alpha}$, $\sigma_{a \alpha} = - \sigma_{a}$ and $\sigma_{b \alpha} = - \sigma_{b}$. With these assumptions, it follows that all three Equations~\eqref{eqn:remain} are satisfied for $a \rightarrow i, b \rightarrow j, \alpha \rightarrow k$. Thus we conclude that Equations~\eqref{eqn:remain} hold whenever $i,j \in M$ and $k \in M^{c}$.
	
	Suppose $m > 1$, we just observed that for $\alpha \in M^{c}$ and $a \in M$ that $\sigma_{\alpha a} = e_{\alpha}$. Thus we can partition the $\alpha \in M^{c}$ by the value $e_{\alpha}$. We consider $\alpha, \beta \in M^{c}$ to be in the same equivalence class, say $I$, if $e_{\alpha} = e_{\beta}$. We define $e_{I}$ such that $\alpha \in I$ implies that $e_{\alpha} = e_{I}$. We denote these equivalence classes by $I$ and $J$.
	
	Suppose $a \in M$ and $\alpha \in I, \beta \in J$. We now check Equations~\eqref{eqn:remain} for $a \rightarrow i, \alpha \rightarrow j, \beta \rightarrow k$. Since $\sigma_{\alpha a}' = \sigma_{\beta a}' = 0$, \cref{eqn:remainA} is satisfied. Equation \eqref{eqn:remainB} and \eqref{eqn:remainC} reduce to the following:
	
	\begin{align}
		\sigma_{\beta \alpha}'(\sigma_{a \alpha}'(e_{\beta}+ \sigma_{a \beta})  - \sigma_{a \beta}'(\sigma_{a \alpha} + e_{\alpha})) & = 0 \label{eq:remainBtemp} \\
		\sigma_{\alpha \beta}'(\sigma_{a \alpha}'(e_{\beta}+ \sigma_{a \beta})  - \sigma_{a \beta}'(\sigma_{a \alpha} + e_{\alpha})) & = 0 \label{eq:remainCtemp}
	\end{align}
	
	Now we can use the fact that $\sigma_{a \alpha} = \sigma_{a \beta} = - \sigma_{a}$ to get the following:
	
	\begin{align}
		\sigma_{\beta \alpha}'(e_{\alpha} - e_{\beta}) & = 0 \\
		\sigma_{\alpha \beta}'(e_{\alpha} - e_{\beta}) & = 0
	\end{align}
	
	If $I = J$ then $e_{\alpha} = e_{\beta}$ and the above equations are satisfied. If $I \neq J$ then we must have that $\sigma_{\alpha \beta}' = \sigma_{\beta \alpha}' = 0$; in which case $\sigma_{\alpha \beta} + \sigma_{\beta \alpha}$ can be absorbed into the $\Phi$ functions. Thus we have proven the following: if $\alpha \in I$ then
	
	\begin{align}
		H_{\alpha}^{2} = \Phi_{\alpha} \prod\limits_{\substack{\beta \in I \\ \beta \neq \alpha}} (\sigma_{\alpha \beta} + \sigma_{\beta \alpha}) \prod\limits_{a =1}^{m} (e_{I} -\sigma_{a}) \quad (m \geq 2)
	\end{align}
	
	Now suppose $m = 1$, then we can pullback to a submanifold given by $x_{1} = \operatorname{constant}$ and then partition the coordinates in $M^{c}$ into connected components. We denote these equivalence classes by $I$ and $J$. Let $a \in M$ and $\alpha \in I, \beta \in J$. As for the case $m > 1$ one can see that \cref{eqn:remainA} is satisfied. Furthermore Equation \eqref{eqn:remainB} and \eqref{eqn:remainC} reduce to Equations \eqref{eq:remainBtemp} and \eqref{eq:remainCtemp} above. If $I \neq J$ then $\sigma_{\beta \alpha}' = \sigma_{\alpha \beta}' = 0$, thus Equations \eqref{eq:remainBtemp} and \eqref{eq:remainCtemp} are both satisfied. Otherwise assume $I = J$ and $\alpha,\beta \in I$ satisfy $\sigma_{\alpha \beta}' \neq 0$ for the moment, then \cref{eq:remainCtemp} implies
	
	\begin{equation}
		\frac{\sigma_{a \alpha}'}{\sigma_{a \beta}'} (e_{\beta}+ \sigma_{a \beta})  - \sigma_{a \alpha} + e_{\alpha} = 0
	\end{equation}
	
	As with \cref{eq:constMultpTemp} we can deduce that $\dfrac{\sigma_{a \alpha}'}{\sigma_{a \beta}'} = d \in \R$. Then the above equation implies that
	
	\begin{equation}
		e_{\beta}+ \sigma_{a \beta} = \frac{e_{\alpha} + \sigma_{a \alpha}}{d}
	\end{equation}
	
	Let $\sigma_{a}^{I} = e_{\alpha} + \sigma_{a \alpha}$, then after absorbing $d$ into the $\Phi$ functions and relabelling, we can assume $\sigma_{a \alpha} = \sigma_{a \beta} = \sigma_{a}^{I}$ and $a_{\alpha a} = a_{\beta b} = 0$. Since the coordinates in $I$ are connected, we can assume that for any $\alpha, \beta \in I$ with $\alpha \neq \beta$ that $\sigma_{a \alpha} = \sigma_{a \beta} = \sigma_{a}^{I}$ and $a_{\alpha a} = a_{\beta b} = 0$ and then Equations \eqref{eq:remainBtemp} and \eqref{eq:remainCtemp} are both satisfied. Then for $\alpha \in I$ and $a \in M$, we have proven the following:
	
	\begin{align}
		H_{\alpha}^{2} = \Phi_{\alpha} \sigma_{a}^{I} \prod\limits_{\substack{\beta \in I \\ \beta \neq \alpha}} (\sigma_{\alpha \beta} + \sigma_{\beta \alpha})\quad (m = 1)
	\end{align}
	
	Also note that Equations~\eqref{eqn:remain} are satisfied whenever $i \in M$ and $j,k \in M^{c}$. Thus we can conclude that Equations~\eqref{eqn:remain} are satisfied whenever at least one of $i$, $j$ or $k$ is in $M$.  Furthermore one can easily check that Equations~\eqref{eqn:remain} are satisfied whenever $i,j,k$ are not all in the same partition.
\end{proof}

When the coordinates are disconnected, i.e. $\sim$ gives multiple partitions, one can easily show that the metric is a product metric.

\begin{proposition} \label{prop:prodMetrics}
	If the coordinates are disconnected, then the metric is a product metric which is given by \cref{eq:prodMetric}. Furthermore Equations~\eqref{eqn:remain} are satisfied whenever i, j or k aren't in the same connected component.
\end{proposition}
\begin{proof}
	Suppose $i \in I$ and $j \in J$ where $I$ and $J$ are a disconnected set of coordinates. Then it follows by definition that $\sigma_{ij'} = \sigma_{ji}' = 0$, thus we can absorb the factors $\sigma_{ij} + \sigma_{ji}$ into the $\Phi$ functions. Thus we have proven that for $i \in I$ the following holds:
	
	\begin{equation}
		H_{i}^{2} = \Phi_{i} \prod\limits_{\substack{j \in I \\ j \neq i}} (\sigma_{ij} + \sigma_{ji})
	\end{equation}
	
	Thus the metric has the form given by \cref{eq:prodMetric}. One can easily check that Equations~\eqref{eqn:remain} are satisfied whenever $i$, $j$ or $k$ aren't in the same connected component.
\end{proof}

\begin{proposition} \label{prop:recurse}
	If $g$ is a reducible orthogonal separable metric with diagonal curvature given by \cref{eq:prodMetric,eq:warpedProdMetric,eq:iregMetric} and $|I| > 1$, the metric $g^{I}$ pulls back to a orthogonal separable metric with diagonal curvature on the submanifold with metric proportional to $g^{I}$.
\end{proposition}
\begin{proof}
	We know that an orthogonally separable web restricted to one of the integral manifolds of its $n$ foliations is still separable \cite{Benenti1993}, and $g^{I}$ still has the form given by \cref{HForm} thus its Riemann curvature tensor will still satisfy $R_{ijik} = 0$ for $j \neq k$ on the integral manifolds.
\end{proof}

We can see how the classification works. If $n = 2$ then we've noted in the previous section that the general solution is given by \cref{eq:SCKTmetricStruct}. So suppose $n > 2$ and the general orthogonal separable metrics with diagonal curvature are known on manifolds with dimension $k \leq n-1$. If the coordinates are disconnected then \cref{prop:prodMetrics} shows us that the metric must have the form given by \cref{eq:prodMetric} and the only equations that haven't been solved are Equations~\eqref{eqn:remain} when $i,j,k$ are inside a connected component. If the coordinates are connected then \cref{prop:ConCoords} in conjunction with \cref{prop:OneConCoord} tells us that the metric must have the form given by \cref{eq:SCKTmetricStruct} or \cref{eq:iregMetric}. Furthermore in this case the only equations that haven't been solved are Equations~\eqref{eqn:remain} when $i,j,k \in I$ where $I \subseteq M^{c}$ is an equivalence class as given in \cref{prop:OneConCoord}. Now fix some $|I| > 1$, then by \cref{prop:recurse}, $g^{I}$ pulls back to an orthogonal separable metric with diagonal curvature on the submanifold with coordinates $(x_{I})$. Thus inductively we know the general form of $g^{I}$ since the dimension of the submanifold is at most $n-1$. In particular if $|I| > 2$ then the components of $g^{I}$ satisfy Equations~\eqref{eqn:remain}. Thus the solution $g$ will satisfy Equations~\eqref{eqn:remain} for all $i,j,k$ distinct and so $g$ satisfies all relevant equations. Thus we have found all the orthogonal separable metrics with diagonal curvature.

\section{Spaces of Constant Curvature} \label{sec:oSepSCC}

In this section our main goal is to show that the metric given \cref{eq:iregMetric} can be ruled out in spaces of constant curvature. In other words, all metrics in spaces of constant curvature are given by \cref{eq:SCKTmetricStruct} and thus are invariantly characterized by concircular tensors.

First we need the following definition. A \emph{twisted product} is a product manifold $M = \prod_{i=0}^{k} M_{i}$ of pseudo-Riemannian manifolds $(M_{i}, g_{i})$ where $\dim M_{i} > 0$ for $i > 0$ equipped with the metric $g = \sum_{i=0}^{k} \rho_{i}^{2} \pi_{i}^{*} g_{i}$ where $\rho_{i} : M \rightarrow \R^{+}$ are functions and $\pi_{i} : M \rightarrow M_{i}$ are the canonical projection maps \cite{Meumertzheim1999}. A \emph{warped product} is a twisted product with $\rho_{0} \equiv 1$ and $\rho_{i} : M_{0} \rightarrow \R^{+}$ for each $i > 0$, and is denoted $M = M_0 \times_{\rho_1} M_1 \times \cdots \times_{\rho_k} M_k$.

Let $M = M_0 \times_{\rho_1} M_1 \times \cdots \times_{\rho_k} M_k$ be a warped product as above. We denote the bilinear product induced by $g$ by $\bp{\cdot,\cdot}$ and $\norm{X}^{2} = \bp{X,X} $ for $X \in T_{p}M$. If $f \in \F(M)$, we denote the Hessian of f \cite[P.~86]{barrett1983semi}, by $S^{f}_{ij} = \nabla_{i}\nabla_{j} f$. If $X,Y \in T_{p}M$ span a non-degenerate 2-plane then the sectional curvature of the 2-plane, $K(X,Y)$, is given in terms of the Riemann curvature tensor $R$ as \cite[lemma~3.39]{barrett1983semi}:

\begin{align}
	K(X,Y) & = \frac{\scalprod{R(X,Y) Y}{X}}{\norm{X \wedge Y}^{2}}, & \norm{X \wedge Y}^{2} & = \bp{X,X}\bp{Y,Y} - \bp{X,Y}^{2}
\end{align}

Denote by $E_i$ the distribution associated with canonical foliation induced by the factor $M_i$. Let $X,Y \in \Gamma(E_0)$, $V \in \Gamma(E_i)$ and $U \in \Gamma(E_k)$ for $i, k > 0$. Let $H_i = - \nabla \log \rho_i$, which is the mean curvature normal of $E_i$ (see \cite{Meumertzheim1999,Rajaratnam2014a}). Then by eq.~(6) in \cite{Meumertzheim1999}, we have the following:

\begin{align}
	K_{XY} & = K_{XY}^{M_0} \label{eq:curv1} \\
	K_{XV} & = - \frac{S^{\rho_{i}}(X,X)}{\rho_{i}\norm{X}^{2}} \label{eq:curv2} \\
	K_{UV} & = - \bp{H_i,H_k} \quad (i \neq k) \label{eq:curv7} \\
	K_{UV} & = \frac{K_{UV}^{M_i} - \norm{\nabla \rho_{i}}^{2}}{\rho^{2}_{i}} \quad (i = k) \label{eq:curv9}
\end{align}

Now suppose that $M$ has constant curvature $\kappa$. After applying the polarization identity to \cref{eq:curv2}, we get

\begin{equation} \label{eq:consCurv2}
	S^{\rho_{i}}(X,Y) + \kappa \rho_{i} \bp{X,Y} = 0
\end{equation}

By \cref{eq:curv9}, we have:

\begin{equation}
	K_{VU}^{M_i} = \kappa \rho^{2}_{i} + \norm{\nabla \rho_{i}}^{2}
\end{equation}

Now,

\begin{align}
	\nabla_{X} K_{VU}^{M_{i}} & = 2 \kappa \rho_{i} \nabla_{X} \rho_{i} + 2 \bp{\nabla_{X} \nabla \rho_{i}, \nabla \rho_{i}} \\
	& = 2(\kappa \rho_{i} \nabla_{X} \rho_{i} + S^{\rho_{i}}(X,\nabla \rho_{i})) \\
	& = 2(\kappa \rho_{i} \nabla_{X} \rho_{i} -\kappa \rho_{i} \bp{X,\nabla \rho_{i}}) \\
	& = 0
\end{align}

Hence for each $i > 0$, $(M_{i}, g_{i})$ necessarily has constant curvature, say $\kappa_{i}$. Finally \cref{eq:curv7} gives us the following:

\begin{equation} \label{eq:consCurv7}
	\bp{\nabla \log \rho_{i}, \nabla \log \rho_{k}} = -\kappa \quad (i \neq k)
\end{equation}

Now suppose $\dim M_0 = 1$ and suppose coordinates on $M_0$ are chosen such that $\tilde{g} = \varepsilon \d x_{1}^{2}$ where $\varepsilon = \pm 1$ as the case may be. Then \cref{eq:consCurv2,eq:consCurv7} imply the following:

\begin{align}
	\partial_{1}^{2} \rho_{i} & = - \kappa \varepsilon \rho_{i}  \\
	(\partial_{1} \log \rho_{i})(\partial_{1} \log \rho_{k}) & = - \varepsilon \kappa \quad (i \neq k)
\end{align}

Hereafter we denote $\omega = \kappa \varepsilon$ and let $\sigma_{i} = \rho_{i}^{2}$. We make exclusive use of the above two equations in the following calculations, suppose $i \neq k$, then:

\begin{align}
	\frac{\sigma_{i}'}{\sigma_{k}'} & = \frac{\rho_{i} \rho_{i}'}{\rho_{k} \rho_{k}'} \\
	& = \frac{\rho_{i} \rho_{i}'}{\rho_{k}^{2}} (-\omega \frac{\rho_{i}'}{\rho_{i}}) \\
	& = - \omega (\frac{\rho_{i}'}{\rho_{k}})^{2}
\end{align}

Thus

\begin{align}
	(\frac{\sigma_{i}'}{\sigma_{k}'})' & = - 2\omega (\frac{\rho_{i}'}{\rho_{k}}) (\frac{\rho_{i}'}{\rho_{k}})' \\
	& = - 2\omega (\frac{\rho_{i}'}{\rho_{k}}) (\frac{\rho_{i}''\rho_{k}-\rho_{i}'\rho_{k}'}{\rho_{k}^{2}}) \\
	& = - 2\omega (\frac{\rho_{i}'}{\rho_{k}^{2}}) (\rho_{i}''+ \omega \rho_{i}) \\
	& = 0
\end{align}

Hence any warped product decomposition of a space of constant curvature with $\dim M_0 = 1$ has a metric given by \cref{eq:warpedProdMetric}. Thus we have proven \cref{thm:KEMsep}.

\section{Conclusion}

In this article we have obtained all orthogonal separable metrics with diagonal curvature. As a consequence, we have generalized some results in \cite{Kalnins1986} to arbitrary spaces of constant curvature and have independently verified  those results for the case of Riemannian spaces of constant curvature.

The results of this article raise the following question: Are there any other spaces of interest in which one can prove that orthogonally separable coordinates have diagonal curvature, or at least admit such coordinates? This is a very nice property to have since if we can rule out the metric \cref{eq:iregMetric}, then the separable coordinates are invariantly characterized by concircular tensors and hence highly amenable to analysis \cite{Crampin2003,Benenti2005a,Rajaratnam2014a}. It may also be of interest to find another condition in addition to the diagonal curvature condition in order to characterize KEM coordinates.

A related question: Given a pseudo-Riemannian manifold, what is a necessary and sufficient intrinsic condition that guarantees the existence of an orthogonal coordinate system having diagonal curvature? One can show that a necessary condition is the existence of an orthogonal coordinate system in which the Ricci tensor is diagonal.

In a coming article, we will classify the point-wise diagonalizable concircular tensors in spaces of constant curvature with Euclidean signature or Lorentzian signature with positive or zero curvature. This will enable one to obtain all orthogonal separable coordinates and apply the BEKM separation algorithm in these spaces \cite{Rajaratnam2014a}. By doing this we will derive necessary and sufficient conditions to construct certain KEM coordinates over an arbitrary warped product.

\section*{Acknowledgments}

We would like to express our appreciation to Dong Eui Chang for his continued interest in this work.  The research was supported in part by National Science and Engineering Research Council of Canada Discovery Grants (D.E.C. and R.G.M.). The first author would like to thank Spiro Karigiannis for reading his thesis \cite{Rajaratnam2014}, which contains the contents of this article.